\documentclass[letterpaper, 10 pt, conference]{ieeeconf}  

\IEEEoverridecommandlockouts                              

\usepackage{graphics} 
\usepackage{epsfig} 
\usepackage{amsmath} 
\usepackage{amssymb}  
\usepackage{xcolor}
\usepackage{nicefrac}
\usepackage{siunitx}
\def\BibTeX{{\rm B\kern-.05em{\sc i\kern-.025em b}\kern-.08em
    T\kern-.1667em\lower.7ex\hbox{E}\kern-.125emX}}

\usepackage{tabularray}
\UseTblrLibrary{amsmath} 
\usepackage{caption}
\usepackage{subcaption}

\usepackage[noconfig,amsmath,thmmarks]{ntheorem}
\theoremnumbering{arabic}
\theoremstyle{plain}
\RequirePackage{latexsym}
\theorembodyfont{\itshape}
\theoremheaderfont{\normalfont\bfseries}
\theoremseparator{}
\newtheorem{assumption}{Assumption}
\newtheorem{remark}{Remark}

\newtheorem{lemma}{Lemma}
\newtheorem{proposition}{Proposition}

\usepackage{tikz}
\usepackage{pgfplots}
\usepackage{cite}
\usepackage{makecell}

\newcommand{\covs}[1]{\mathrm{cov}\left( #1 \right)}
\newcommand{\E}[1]{\mathrm{E}\left\{#1\right\}}

\newcommand{\bigA}{\boldsymbol{A}}
\newcommand{\bigB}{\boldsymbol{B}}
\newcommand{\bigH}{\boldsymbol{H}}
\newcommand{\bigP}{\boldsymbol{P}}
\newcommand{\bigp}{\boldsymbol{p}}
\newcommand{\bigK}{\boldsymbol{K}}
\newcommand{\bigL}{\boldsymbol{L}}

\newcommand{\bigOmega}{\boldsymbol{\Omega}}
\newcommand{\bigSigma}{\boldsymbol{\Sigma}}
\newcommand{\bigx}{\boldsymbol{x}}
\newcommand{\bigy}{\boldsymbol{y}}
\newcommand{\bigu}{\boldsymbol{u}}
\newcommand{\bigxi}{\boldsymbol{\xi}}
\newcommand{\bigomega}{\boldsymbol{\omega}}

\newcommand{\bigepsilon}{\boldsymbol{\epsilon}}
\newcommand{\bigvarepsilon}{\boldsymbol{\varepsilon}}
\newcommand{\bigR}{\boldsymbol{R}}
\newcommand{\bigQ}{\boldsymbol{Q}}

\newcommand{\bigmathA}{\boldsymbol{\mathcal{A}}}

\newcommand{\bigC}{\boldsymbol{C}}
\newcommand{\bigD}{\boldsymbol{D}}
\newcommand{\bigZ}{\boldsymbol{Z}}
\newcommand{\bige}{\boldsymbol{e}}
\newcommand{\bigalpha}{\boldsymbol{\alpha}}
\newcommand{\bigbeta}{\boldsymbol{\beta}}
\newcommand{\bigpsi}{\boldsymbol{\psi}}

\usepackage{ifthen}
\newboolean{showcomments}
\setboolean{showcomments}{true} 
\ifthenelse{\boolean{showcomments}}{ 
	\newcommand\td[1]{\textcolor{red}{\textbf{TODO:} #1}} 
	\newcommand\tdd[1]{} 
	\newcommand\comment[1]{\textcolor{blue}{\textbf{Comment:} #1}} 
	\newcommand\commentd[1]{} 
	\newcommand\frage[1]{\textcolor{orange}{\textbf{Rückfrage:} #1}} 
	\newcommand\fraged[1]{}
}{ 
	\newcommand\td[1]{}
	\newcommand\tdd[1]{} 
	\newcommand\comment[1]{} 
	\newcommand\commentd[1]{} 
	\newcommand\frage[1]{}
	\newcommand\fraged[1]{}
} 

\newcommand\copyrighttext{%
  \footnotesize \textcopyright This work has been submitted to the IEEE for possible publication. Copyright may be transferred without notice, after which this version may no longer be accessible.}
\newcommand\copyrightnotice{%
\begin{tikzpicture}[remember picture,overlay]
\node[anchor=south,yshift=10pt] at (current page.south) {\fbox{\parbox{\dimexpr\textwidth-\fboxsep-\fboxrule\relax}{\copyrighttext}}};
\end{tikzpicture}%
}

\title{\LARGE \bf
Human-Variability-Respecting Optimal Control for Physical Human-Machine Interaction
}

\author{Sean Kille, Paul Leibold, Philipp Karg, Balint Varga, Sören Hohmann
\thanks{All authors are with the Institute of Control Systems (IRS), Karlsruhe Institute of Technology (KIT), 76131 Karlsruhe, Germany. Corresponding author is {\tt\small sean.kille@kit.edu}.}%
}

\begin{document}

\maketitle
\copyrightnotice

\begin{abstract}

Physical Human-Machine Interaction plays a pivotal role in facilitating collaboration across various domains. 
When designing appropriate model-based controllers to assist a human in the interaction, the accuracy of the human model is crucial for the resulting overall behavior of the coupled system. 
When looking at state-of-the-art control approaches, most methods rely on a deterministic model or no model at all of the human behavior. 
This poses a gap to the current neuroscientific standard regarding human movement modeling, which uses stochastic optimal control models that include signal-dependent noise processes and therefore describe the human behavior much more accurate than the deterministic counterparts.
To close this gap by including these stochastic human models in the control design, we introduce a novel design methodology resulting in a \textit{Human-Variability-Respecting Optimal Control} that explicitly incorporates the human noise processes and their influence on the mean and variability behavior of a physically coupled human-machine system.
Our approach results in an improved overall system performance, i.e. higher accuracy and lower variability in target point reaching, while allowing to shape the joint variability, for example to preserve human natural variability patterns. 
\end{abstract}
\section{INTRODUCTION} \label{sec:introduction}

Physical Human-Machine Interaction is becoming increasingly relevant as it focuses on machines that intuitively interact with and support humans in various applications. Such applications are e.g. construction tasks that are too complex to be fully automated or mobility-assisting or rehabilitating humans that recover from injuries, thereby playing a crucial role in the future of manual work~\cite{Carlson.2012}. 
Developments in these fields are strongly supported by research in shared control, which focuses on enhancing machine automation which continuously and parallel to the human interacts with a system haptically~\cite{Abbink.2012}. 
This aims at creating an improved interaction experience in which human and machine combine their strengths~\cite{Inga.2023}. One major challenge in shared control design lies in the task of parametrizing the automation, which is either done heuristically or based on a model of the human behavior.


Early shared control designs, particularly in the automotive domain, employ heuristically designed feedback controllers without detailed consideration of human behavior~\cite{Steele.2001, Flemisch.2008}. 
These solutions proved how haptic assistance can support the human in its task accomplishment while still keeping the human in the loop. 
However, the designs need heuristic tuning which takes time and is task-specific.
Approaches focusing on object manipulation or supporting hand movements often adapt the overall system impedance based on modeling humans as impedance systems~\cite{Lecours.2012, Tonietti.2005, Dong.2020, Ficuciello.2015, Li.2017}. 
These solutions lead to improved performance but usually adapt the impedance of a passive system, thereby missing an active assistance support by the automation. 
Other approaches involve game-theoretic controllers, which model human motion based on optimality principles~\cite{Franceschi.2023b, Flad.2017, Music.2020, Na.2013, Varga.2024}. 
These models are based on observations from movement sciences and neuroscience~\cite{Gallivan.2018,Todorov.2002}, leading to generalizable models on which an automation can be designed. 


All these mentioned designs assume a deterministic human behavior. 
However, neuroscientific literature shows that the human behaves stochastically~\cite{Abend.1982, Harris.1998} and adheres to the minimal intervention principle~\cite{Todorov.2002,Todorov.2004}, resulting in a reduced variability only in areas that are relevant for a successful task-completion. 
This pattern of a high variability in areas that are not of interest for the human task performance is termed task-dependent variability and is characteristic for human movements. 
The state-of-the-art linear-quadratic sensorimotor (LQS) model represents this variability~\cite{Todorov.2005}. The model includes additive and multiplicative noise processes in the modeled human action and perception cycle. This leads to a better accuracy in describing both the mean behavior and variability patterns~\cite{Karg.2023, Mitrovic.2011}. The reason for this superiority is that noise parameters not only affect variability, but also the mean behavior~\cite{Todorov.2005,Karg.2022}.

When looking at the state-of-the-art shared control design approaches, two observations stand out: A)~The effect of the human noise processes on the mean behavior is not reflected, leading to a suboptimal performance of model-based shared control designs as they are based on deterministic human models. B)~The control designs do not consider the inherent human variability but often interpret it as undesired noise, therefore urging the human to behave deterministically and thereby restricting the human in its natural high variability in task-irrelevant areas.  

 
This paper adresses these two research gaps by introducing a novel approach to design a controller which aims at assisting a human in physical Human-Machine Interaction. This \textit{Human-Variability-Respecting Optimal Control (HVROC)} approach implements and parametrizes an optimal controller as an automation. Research gap A is adressed by  explicitly taking into account the joint mean behavior of a stochastically modeled human and automation in an optimization-based parametrization.
As a basis for this approach and as a further contribution, we introduce the recursive calculation of the mean behavior and variance of a coupled stochastic human-machine system. Research gap B is tackled by introducing a variability-pattern shaping of the joint variance.
The resulting controller's benefits are twofold:
First of all, our \textit{HVROC} considers the effect of the human noise processes on the joint mean behavior, enhancing the overall system performance in task accuracy, speed and goal point variability.
Secondly, our design enables to influence the resulting variability a) either to mimic the high human natural variability in task-irrelevent areas b) or restrict it. Such an adaption is possible due to the newly formulated relationship between the automation and the resulting mean and variance of the coupled system.
The benefits of our approach are shown in simulation using two example systems.

\section{Human-Variability-Respecting Optimal Control} 

In this section we firstly introduce our system model in Subsection~\ref{sec:system}, which is influenced by a human and an automation. 
The stochastic human model is described in Subsection~\ref{sec:humanModel} before we present our novel automation design in Subsection~\ref{subsec:HVROC}. We thereby focus on point-to-point (P2P) movements without loss of generality.  


\subsection{Human-Machine System Model}\label{sec:system}
We define a system  which is manipulated by both a human and an automation, thereby adhering to the current standart representation of human biomechanical movement modeling by~\cite{Todorov.2005}:
\begin{align}
\bigx_{t+1}= \bigA \bigx_t + \bigB_{\mathrm{A}} \bigu_{\mathrm{A},t} +\bigB_{\mathrm{H}} \bigu_{\mathrm{H},t}  + \bigxi_t + \sum_i^c \varepsilon_t^{(i)} \bigC_i \bigu_{\mathrm{H},t}, \label{eq:system}
\end{align}
where $\bigx\in\mathbb{R}^{n}$ denotes the system state, $\bigu_{\mathrm{H}}\in\mathbb{R}^{m_H}$ and $\bigu_{\mathrm{A}}\in\mathbb{R}^{m_A}$ the control variables of the human and automation, respectively. 
Additionally, we encounter an additive standard white Gaussian noise process with $\bigxi_t \in\mathbb{R}^n$ and the human control-dependent noise process $\sum_i \varepsilon_t^{(i)} \bigC_i \bigu_{\mathrm{H},t}\in\mathbb{R}^n$. The latter can be interpreted as a higher inaccuracy for faster movements. 
With $\bigA, \bigB_{\mathrm{A}},\bigB_{\mathrm{H}}, \bigC_i$ we introduce matrices of appropriate dimension and $\bigvarepsilon = \begin{bmatrix} \varepsilon_t^{(1)} \dots \varepsilon_t^{(c)} \end{bmatrix}^\intercal$ denotes a standard white Gaussion noise process ($\covs{\bigvarepsilon} = \boldsymbol{I}$).
The stochastic process ${\bigx_t}$ is initialized with $\mathrm{E}\{\bigx_0\}$ and $\mathrm{cov}(\bigx_0, \bigx_0) = \bigOmega_{0}^{\bigx}$. 
The automation typically does not have full access to the system states; it's output equation can be described as:
\begin{align}
    \bigy_{\mathrm{A},t} &= \bigH_{\mathrm{A}} \bigx_t, \label{eq:sysOutA}
\end{align}
while the human perception is complemented by the additive $\bigomega_t$ and multiplicative noise process $ \sum_i \epsilon_t^{(i)} \bigD_i \bigx_t$ in the output equation:
\begin{align}
        \bigy_{\mathrm{H},t} &= \bigH_{\mathrm{H}} \bigx_t + \bigomega_{t} + \sum_{i}^{d} \epsilon_t^{(i)} \bigD_i \bigx_t, \label{eq:sysOutH}
\end{align}
where $\bigy_{\mathrm{H}}\in\mathbb{R}^{r_H}$ and $\bigy_{\mathrm{A}}\in\mathbb{R}^{r_A}$ denote the observed output by the human and the automation, respectively. 
The state-dependent noise translates to a higher inaccuracy in perception, the faster a movement is.
$ \bigomega$ and $\bigepsilon = \begin{bmatrix} \epsilon_t^{(1)}  \dots \epsilon_t^{(d)} \end{bmatrix}^\intercal$ with $\covs{\bigepsilon}=\boldsymbol{I}$ denote standard white Gaussian noise processes.
$\bigH_{\mathrm{A}}$, $\bigH_{\mathrm{H}}$ and $\bigD_i$ are matrices of appropriate dimension. 
All mentioned Gaussian noise processes are independent to each other and to $\bigx_t$.
The additive noise processes in human action and perception are composed of standard white Gaussian noise processes $\bigalpha_t \in \mathbb{R}^p$ and $\bigbeta_t \in \mathbb{R}^q$  ($\covs{\bigalpha, \bigalpha} = \covs{\bigalpha} =\boldsymbol{I}$ and $\covs{\bigbeta} =\boldsymbol{I}$): $\bigxi_t = \bigSigma^{\bigxi} \bigalpha_t$ and $\bigomega_t = \bigSigma^{\bigomega} \bigbeta_t$.


\subsection{Human Behavior}\label{sec:humanModel}

When regarding only the human actor in our system dynamics~\eqref{eq:system}, i.e. $\bigu_\mathrm{A}=\boldsymbol{0}$, the linear-quadratic sensorimotor (LQS) model for goal-directed human movements as introduced by~\cite{Todorov.2005} results. 
According to that, the human chooses $\bigu_{\mathrm{H},t}$ such that the performance criterion 
\begin{align}\label{eq:costHuman}
    J_\mathrm{H} \!=\!  \E{
    \bigx_N^\intercal \bigQ_{\mathrm{H},N}\bigx_{N} \! +\! \sum_{t=0}^{N-1} \left( \bigx_t^\intercal \bigQ_{\mathrm{H},t} \bigx_t \!+\! \bigu_{\mathrm{H},t}^\intercal \bigR_{\mathrm{H}} \bigu_{\mathrm{H},t} \right) },
\end{align}
becomes minimal, with $\bigQ_{\mathrm{H},t} \in \mathbb{R}^{n \times n}$ being symmetric and positiv semidefinite for $t = 0,...,N-1$ and $\bigR_\mathrm{H} \in \mathbb{R}^{m \times m}$ being symmetric and positive definite. 
Eq.~\eqref{eq:costHuman} can be transformed into a form that allows penalizing a deviation from a reference state (i.e. $\bigx_t - \bigx_\mathrm{ref}$) by introducing additional position states $p_{x,\mathrm{ref}}$ and $p_{y,\mathrm{ref}}$ with constant dynamics~\cite{Todorov.2002}, thus enabling the analyzis of P2P movements from a start point $\bigp_0$ to an end point $\bigp_{\mathrm{ref}}$ with $\bigp_t$ denoting the position state.

An approximate solution to this optimal control problem is presented in~\cite{Todorov.2005}, resulting in the control law $\bigu_{\mathrm{H},t} = - \bigL_{\mathrm{H},t} \hat{\bigx}_{\mathrm{H},t}$. 
This means that the human acts on the system based on its estimated system state $\hat{\bigx}_{\mathrm{H},t}$ and the control matrix $\bigL_{\mathrm{H},t}$.
The state estimation is updated recursively by a linear Kalman filter $\hat{\bigx}_{\mathrm{H},t+1}=\bigA \hat{\bigx}_{\mathrm{H},t}+ \bigB_\mathrm{H} \bigu_{\mathrm{H},t}+\bigK_{\mathrm{H},t}\left(\bigy_{\mathrm{H},t}- \bigH_{\mathrm{H}} \hat{\bigx}_{\mathrm{H},t}\right)$.  
Adhering to the derivation by~\cite{Todorov.2005} and~\cite{Karg.2022}, the approximate solution for control and filter matrix can be found by iterating between 
\begin{align}
   \bigL_{\mathrm{H},t} =& \Big( \bigR_{\mathrm{H}} + \bigB_{\mathrm{H}}^\intercal \bigZ^{\bigx_\mathrm{H}}_{t+1} \bigB_{\mathrm{H}}  \notag \\
   &+  \sum_{i}^{} \bigC_i^\intercal ( \bigZ^{\bigx_\mathrm{H}}_{t+1}+\bigZ^{\boldsymbol{e}_\mathrm{H}}_{t+1} ) \bigC_i  \Big)^{-1} \bigB_{\mathrm{H}}^\intercal \bigZ^{\bigx_\mathrm{H}}_{t+1} \bigA \label{eq:L_H}
\end{align} 
and
\begin{align}
    \bigK_{\mathrm{H},t} =& \bigA \bigP_t^{\bige_\mathrm{H}} \bigH_\mathrm{H}^\intercal \Big( \bigH_\mathrm{H} \bigP_t^{\bige_\mathrm{H}} \bigH_\mathrm{H} + \bigOmega^{\bigomega} \notag \\
    &+ \sum_{i}^{}\! \bigD_i ( \bigP_t^{\bige_\mathrm{H}}\! +\! \bigP_t^{\hat{\bigx}_\mathrm{H}}\! + \! \bigP_t^{\hat{\bigx}_\mathrm{H}  \bige_\mathrm{H}} \! + \! \bigP_t^{\bige_\mathrm{H} \hat{\bigx}_\mathrm{H}  }  ) \bigD_i^\intercal \Big)^{-1}, \label{eq:K_H}    
\end{align}
with the recursive equations for $\bigZ^{\bigx_\mathrm{H}}_{t}$, $\bigZ^{\boldsymbol{e}_\mathrm{H}}_{t}$, $ \bigP_t^{\bige_\mathrm{H}}$, $\bigP_t^{\hat{\bigx}_\mathrm{H}}$, $\bigP_t^{\hat{\bigx}_\mathrm{H} \bige_\mathrm{H}}$, $\bigP_t^{\bige_\mathrm{H} \hat{\bigx}_\mathrm{H}}$ being provided by~\cite{Todorov.2005} and $\bigOmega^{\bigomega} = \covs{\bigomega}$.

According to Lemma~2 in~\cite{Karg.2022}, when applying the optimal control strategy to the system described by~\eqref{eq:system} and~\eqref{eq:sysOutH} with $\bigu_\mathrm{A}=\boldsymbol{0}$, the mean $\E{\bigx_t}$ and variance $\bigOmega_{t}^{\bigx}$ can be described as follows:  
\begin{align}
    \left[\begin{array}{l}
        \mathrm{E}\left\{\boldsymbol{x}_{t+1}\right\} \\
        \mathrm{E}\left\{\hat{\boldsymbol{x}}_{\mathrm{H},t+1}\right\}
        \end{array}\right]  = & \boldsymbol{\mathcal{A}}_{\mathrm{H},t}\left[\begin{array}{l}
        \mathrm{E}\left\{\boldsymbol{x}_t\right\} \\
        \mathrm{E}\left\{\hat{\boldsymbol{x}}_{\mathrm{H},t}\right\}
        \end{array}\right], \label{eq:E_human}\\
    \mathrm{cov}
    \begin{pmatrix}
        \begin{bmatrix}
        \bigx_{t+1} \\
        \hat{\bigx}_{\mathrm{H},t+1}
        \end{bmatrix}  
    \end{pmatrix} =
    & \boldsymbol{\mathcal{A}}_{\mathrm{H},t} 
    \mathrm{cov}
    \begin{pmatrix}
        \begin{bmatrix}
            \bigx_{t} \\
            \hat{\bigx}_{\mathrm{H},t}
        \end{bmatrix}  
    \end{pmatrix}
  \boldsymbol{\mathcal{A}}_{\mathrm{H},t}^{\intercal} \notag \\
    & +\left[\begin{array}{cc}
    \bigOmega^{\boldsymbol{\xi}} & \mathbf{0} \\
    \boldsymbol{0} & \bigK_{\mathrm{H},t} \boldsymbol{\bigOmega}^{\omega} \boldsymbol{K}_{\mathrm{H},t}^{\intercal}
    \end{array}\right] \notag \\
    & +\left[\begin{array}{cc}
    \bar{\bigOmega}_t^{\hat{\bigx}_\mathrm{H}} & \mathbf{0} \\
    \mathbf{0} & \bar{\bigOmega}_t^{\bigx}
    \end{array}\right] \label{eq:cov_human}
\end{align}  
with $\mathrm{cov}(\boldsymbol{\epsilon}) = (\mathrm{cov}(\epsilon_{i}, \epsilon_{j}))_{i,j=1,...,n} $ for a multivariate random variable $\boldsymbol{\epsilon} \in \mathbb{R}^{n \times 1}$, $\boldsymbol{\Omega}_{t}^{\boldsymbol{x}}=\covs{\boldsymbol{x}_t}$ for a variable $\boldsymbol{x}$, 
$\bar{\bigOmega}^{\hat{\bigx}_{\mathrm{H}}}_t = \sum_i \bigC_i\bigL_{\mathrm{H},t}, \left(\bigOmega^{\hat{\bigx}_{\mathrm{H}}}_t + \mathrm{E}\{\hat{\bigx}_{\mathrm{H},t}\} \mathrm{E}\{\hat{\bigx}_{\mathrm{H},t}\}^\intercal \right)\bigL_{\mathrm{H},t}^{\intercal} \bigC_i^\intercal$,
$\bar{\bigOmega}^{\bigx}_t =  \sum_i \bigK_{\mathrm{H},t} \bigD_i \left(\bigOmega^{\bigx}_t + \mathrm{E}\{\bigx_t\} \mathrm{E}\{\bigx_t\}^\intercal \right)  \bigD_i^\intercal \bigK_{\mathrm{H},t}^\intercal,$
\begin{align}
    \boldsymbol{\mathcal{A}}_{\mathrm{H},t} = 
    \begin{bmatrix}
        \bigA & \! - \! \bigB_{\mathrm{H}}\bigL_{\mathrm{H},t} \\
       \bigK_{\mathrm{H},t} \bigH_{\mathrm{H}} & \bigA \! - \! \bigB_{\mathrm{H}}\bigL_{\mathrm{H},t} \! - \! \bigK_{\mathrm{H},t} \bigH_{\mathrm{H}} \\
       \end{bmatrix}
\end{align}
and initial values $\mathrm{E}\{\hat{\bigx}_{\mathrm{H},0}\}= \hat{\bigx}_{\mathrm{H},0} = \mathrm{E}\{\bigx_{0}\}$ and 
\begin{align}
    \mathrm{cov}
    \begin{pmatrix}
        \begin{bmatrix}
        \bigx_{0} \\
        \hat{\bigx}_{\mathrm{H},0}
        \end{bmatrix}  
    \end{pmatrix} =
    \left[\begin{array}{cc}
        \bigOmega_0^{\bigx} & \mathbf{0} \\
        \mathbf{0} & \mathbf{0}
        \end{array}\right].
\end{align}

\subsection{Human Variability-Respecting Optimal Control}\label{subsec:HVROC}

\subsubsection{Automation} \label{sec:automation}

As an automation we use a LQ optimal controller. 
The performance criterion for the automation is defined as: 

\begin{align}\label{eq:costAutomation}
    J_\mathrm{A} =  
    \bigx_N^\intercal \bigQ_{\mathrm{A},N}\bigx_{N} + \sum_{t=0}^{N-1} \left( \bigx_t^\intercal \bigQ_{\mathrm{A},t} \bigx_t + \bigu_{\mathrm{A},t}^\intercal \bigR_{\mathrm{A}} \bigu_{\mathrm{A},t} \right), 
\end{align}
with $\bigQ_{\mathrm{A},t} = \mathrm{diag}(\boldsymbol{q}_\mathrm{A})$ for $t = 0,...,N-1$ and $\bigR_{\mathrm{A}} = \mathrm{diag}(\boldsymbol{r}_\mathrm{A})$ and with $\boldsymbol{q}_\mathrm{A} \in \mathbb{R}^{n \times 1}$ being positive semidefinite and $\boldsymbol{r}_\mathrm{A} \in \mathbb{R}^{m_\mathrm{A} \times 1}$ being positive definite. 
This results in a control behavior that is similar to the human and aims at mimicking the single human behavior.

Eq.~\eqref{eq:costAutomation} can be augmented like in the human case, such that it allows a penalization of the deviation from a reference state $\bigp_{\mathrm{ref}}$.
Similar to the human but deterministic case, a control law and state estimation can be derived: 
\begin{align}
    \bigu_{\mathrm{A},t} = & - \bigL_{\mathrm{A},t} \hat{\bigx}_{\mathrm{A},t} \label{eq:u_aut}\\
    \hat{\bigx}_{\mathrm{A},t+1} = & \bigA \hat{\bigx}_{\mathrm{A},t} + \bigB_\mathrm{A} \bigu_{\mathrm{A},t} \notag \\ 
    &+\bigK_{\mathrm{A},t}(\bigy_{\mathrm{A},t}- \bigH_{\mathrm{A}} \hat{\bigx}_{\mathrm{A},t}). \label{eq:E_aut}
\end{align}  
These result from the iterative calculation of:
\begin{align}
    \bigL_{\mathrm{A},t} =& \Big( \bigR_{\mathrm{A}} + \bigB_{\mathrm{A}}^\intercal \bigZ_{t+1} \bigB_{\mathrm{A}} \Big)^{-1} \bigB_{\mathrm{A}}^\intercal \bigZ_{t+1} \bigA \label{eq:L_A}, \\
     \bigK_{\mathrm{A},t} =& \bigA \bigP_t \bigH_\mathrm{A}^\intercal \Big( \bigH_\mathrm{A} \bigP_t \bigH_\mathrm{A}^\intercal \Big)^{-1}, \label{eq:K_A}    
 \end{align}
with the equations for $\bigZ$ and $\bigP$ being provided by e.g.~\cite{Anderson.1989}.


The overall system architecture, picturing the interplay between the automation, system model (\ref{sec:system}) and the human model (\ref{sec:humanModel}) is depicted in Fig.~\ref{fig:system}.

\begin{figure*}[thpb]
    \centering
    \resizebox{0.9\textwidth}{!}{
          \definecolor{shadecolor}{rgb}{0.92,0.92,0.92}

\tikzstyle{block} = [rectangle, draw, text width=2cm, text centered, minimum height=1.5cm]
 \tikzstyle{box} = [rectangle, draw, fill=shadecolor, text width=9em, text centered, rounded corners, minimum height=4em]
\tikzstyle{noise} = [rectangle, draw, text width=1.5cm, text centered, minimum height=0.8cm]
\tikzstyle{noiseS} = [rectangle, draw, text width=1cm, text centered, minimum height=0.8cm]
\tikzstyle{plus} = [circle, draw, text width=0.1cm, text centered]
 \tikzstyle{frame} = [rectangle, dashed, draw, text width=9em, text centered, minimum height=4em]

\tikzstyle{line} = [draw, -latex']

\usetikzlibrary{arrows}
\begin{tikzpicture}
 
\node[block, fill=shadecolor, text width=5cm, text centered, minimum height=2cm] (system) at (3.5,1.75) {Biomechanics \\ $\boldsymbol{x}_{t+1}= \boldsymbol{A} \boldsymbol{x}_t + \boldsymbol{z}_{\mathrm{H},t} + \boldsymbol{z}_{\mathrm{A},t}$ };

\node[box,  inner xsep = 5.5cm, inner ysep=2.75cm] (boxH) at (-7.5,4.5) {};
 \node [draw=none, anchor=north west, align = right] at (boxH.north west) {Human};

\node[block,  text width=4.5cm] (estimateH) at (-9.5,5.5) {Estimation \\ $\hat{\boldsymbol{x}}_{\mathrm{H},t+1}=\boldsymbol{A} \hat{\boldsymbol{x}}_{\mathrm{H},t}+ \boldsymbol{B}_\mathrm{H} \boldsymbol{u}_{\mathrm{H},t}$ \\ $+\boldsymbol{K}_{\mathrm{H},t}\left(\boldsymbol{y}_{\mathrm{H},t}- \boldsymbol{H}_{\mathrm{H}} \hat{\boldsymbol{x}}_{\mathrm{H},t}\right)$};
\node[block, text width=3cm] (controlH) at (-12.5,3.5) {Control law \\ $\boldsymbol{u}_{\mathrm{H},t} = - \boldsymbol{L}_{\mathrm{H},t} \hat{\boldsymbol{x}}_{\mathrm{H},t} $};

\node[plus] (plusPercept) at (-5,5.5) {};
\path[line] (plusPercept) -- (estimateH)node[midway, xshift=0.25cm, yshift=0.25cm](yH){$\boldsymbol{y}_{\mathrm{H},t}$};
\path[line] (estimateH) -- (controlH.north |- estimateH)node[midway, yshift=0.25cm, xshift=-0.25cm](hatxH){$\hat{\boldsymbol{x}}_{\mathrm{H},t}$} -- (controlH);  
\path[line] (controlH) -- (estimateH.south |- controlH)node[midway, yshift=0.25cm, xshift=0cm](uH){$\boldsymbol{u}_{\mathrm{H},t}$} -- (estimateH);  

\node (aidSys) at (6.5,1.75) {};
\node[noise] (H_H) at (-3,5.5) {$\boldsymbol{H}_\mathrm{H}$};

\path[line] (system) -- (aidSys.center) -- (H_H.east -| aidSys.center) -- (H_H);%
\path[line] (H_H) -- (plusPercept);

\node[noise] (mulPercept) at (-3,6.5) {$\sum_i \epsilon_t^i \boldsymbol{D}_i $};
\node[noiseS] (addPercept) at (-5,6.5) {$\boldsymbol{\omega}_t$};

\node (aidNPIn) at (-0.75,5.5) {};
\path[line] (aidNPIn.center) -- (mulPercept.east -| aidNPIn.center) -- (mulPercept);

\path[line] (mulPercept.west) -- (plusPercept)node[yshift=0cm](plus1){$+$};
\path[line] (addPercept.south) -- (plusPercept);

\node[plus] (plusControl) at (-2.5,3.5) {};

\node[noise] (B_H) at (-5,3.5) {$\boldsymbol{B}_\mathrm{H}$};
\node[noise] (mulControl) at (-5,2.5) {$\sum_i \varepsilon_t^i \boldsymbol{C}_i $};
\node[noiseS] (addControl) at (-2.5,2.5) {$\boldsymbol{\xi}_t$};
\path[line] (controlH) -- (B_H);
\path[line] (B_H) -- (plusControl)node[yshift=0cm](plus2){$+$};
\path[line] (mulControl.east) -- (plusControl);
\path[line] (addControl.north) -- (plusControl);

\node (aidNCIn) at (-6.25,3.5) {};
\path[line] (aidNCIn.center) -- (mulControl.west -| aidNCIn.center) -- (mulControl);

\node (aidH) at (0.5,3.5) {};
\node (sysInH) at (1,2.0) {};

\path[line] (plusControl)-- (aidH.center)node[midway, xshift=1.25cm, yshift=0.25cm](zH){$\boldsymbol{z}_{\mathrm{H},t}$} -- (sysInH.west -| aidH.center) -- (sysInH);


\node[box,  inner xsep = 5.5cm, inner ysep=2.25cm] (boxA) at (-7.5,-0.75) {};
 \node [draw=none, anchor=north west, align = right] at (boxA.north west) {Automation};
  
\node[block,  text width=4.5cm] (estimateA) at (-8.5,-2) {Estimation \\ $\hat{\boldsymbol{x}}_{\mathrm{A},t+1}=\boldsymbol{A} \hat{\boldsymbol{x}}_{\mathrm{A},t}+ \boldsymbol{B}_\mathrm{A} \boldsymbol{u}_{\mathrm{A},t} $\\$+\boldsymbol{K}_{\mathrm{A},t}\left(\boldsymbol{y}_{\mathrm{A},t}- \boldsymbol{H}_{\mathrm{A}} \hat{\boldsymbol{x}}_{\mathrm{A},t}\right)$};
\node[block] (perceptA) at (-2.5,-2) {Perception $\boldsymbol{H}_\mathrm{A}$};

\node[block, text width=3.5cm] (controlA) at (-12,0) {Control law \\ $\boldsymbol{u}_{\mathrm{A},t} = - \boldsymbol{L}_{\mathrm{A},t} \hat{\boldsymbol{x}}_{\mathrm{A},t} $};

\node[block] (B_A) at (-2.5,0) {Action \\ $\boldsymbol{B}_\mathrm{A}$};

\path[line] (perceptA) -- (estimateA)node[midway, yshift=0.25cm](yA){$\boldsymbol{y}_{\mathrm{A},t}$};
\path[line] (estimateA) -- (controlA.north |- estimateA)node[midway, yshift=0.25cm](hatxA){$\hat{\boldsymbol{x}}_{\mathrm{A},t}$} -- (controlA);  
\path[line] (controlA) -- (estimateA.south |- controlA)node[midway, yshift=0.25cm](uA){$\boldsymbol{u}_{\mathrm{A},t}$} -- (estimateA);


\path[line] (system) -- (aidSys.center)node[midway, xshift=0.5cm](yA){$\boldsymbol{x}_t$} -- (perceptA.east -| aidSys.center) -- (perceptA);

\node (aidA) at (0.5,0) {};

\node (sysInA) at (1,1.5) {};

\node (sysInH) at (1,1.0) {};

\path[line] (controlA) -- (B_A);
\path[line] (B_A)-- (aidA.center)node[midway, xshift=1cm, yshift=-0.25cm](zH){$\boldsymbol{z}_{\mathrm{A},t}$} -- (sysInA.west -| aidA.center) -- (sysInA);


\node[frame,  inner xsep = 1.25cm, inner ysep=1.25cm] (framePercept) at (-3.5,5.8) {};
 \node [draw=none, anchor=south west, align = right] at (framePercept.south west) {Perception};

\node[frame,  inner xsep = 1.25cm, inner ysep=1.25cm] (frameAction) at (-3.5,3.2) {};
 \node [draw=none, anchor=north east, align = right] at (frameAction.north east) {Action};

\end{tikzpicture}
      }
      \caption{Human-Machine System Model}
      \label{fig:system}
\end{figure*}
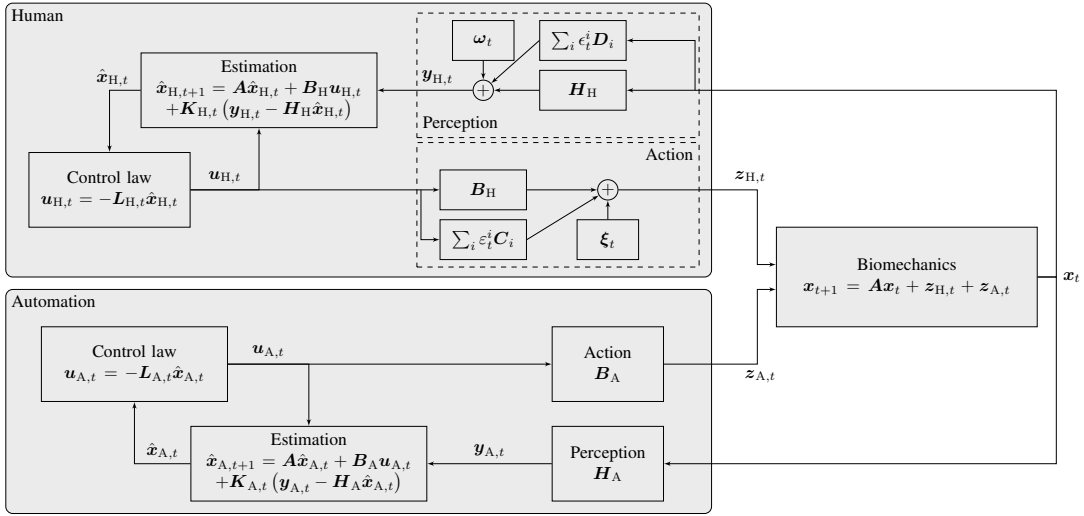



\subsubsection{Control Objectives} \label{sec:goal}
As introduced in Section~\ref{sec:introduction}, we assume that maintaining human movement characteristics, i.e. high variability in task-irrelevant areas, in a joint interaction between a human and an automation leads to an enhanced human experience. 
Additionally, we expect an improved overall performance by considering the effect of the human noise processes on the mean behavior.
Concretely, with our \textit{HVROC}  design we can find an automation that, compared to a human-only action, leads to:
\begin{enumerate}
    \item a) a similarly high or b) strongly constricted variance in task-irrelevant areas,
    \item a smaller variance in reaching the goal point,
    \item a  smaller mean error in reaching the goal point
\end{enumerate}
in a joint interaction on the system. 
These goals can be formalized to:
\begin{align} 
    \mathrm{1a)} & \quad  \bigOmega_{\nicefrac{N}{2},\mathrm{HVROC}}^{\bigp} \approxeq \bigOmega_{\nicefrac{N}{2},\mathrm{human}}^{\bigp}, \label{eq:goal1a} \\ 
    \mathrm{1b)} & \quad  \bigOmega_{\nicefrac{N}{2},\mathrm{HVROC}}^{\bigp} \ll \bigOmega_{\nicefrac{N}{2},\mathrm{human}}^{\bigp}, \label{eq:goal1b} \\
    \mathrm{2)} & \quad \bigOmega_{N,\mathrm{HVROC}}^{\bigp} \leq \bigOmega_{N,\mathrm{human}}^{\bigp}, \label{eq:goal2} \\
    \mathrm{3)} & \quad \vert \mathrm{E}\{ \boldsymbol{p}_N \}_\mathrm{HVROC} - \boldsymbol{p}_{\mathrm{ref}} \vert \leq \vert \mathrm{E}\{ \boldsymbol{p}_N \}_\mathrm{human} - \boldsymbol{p}_{\mathrm{ref}} \vert. \label{eq:goal3}
\end{align}
The suffix $\mathrm{human}$ and $\mathrm{HVROC}$ denote whether the human alone or both human and controller perform the task. 
The task-irrelevant area in our task is defined as the midway area between the starting point $\bigp_0$ and endpoint $\bigp_\mathrm{ref}$. 
Therefore we measure the variance of the task-irrelevant area at the time step $t=\nicefrac{N}{2}$.

Our work is based on the following assumption: 
\begin{assumption} \label{assumption1} 
    All parameters describing the human behavior, i.e. cost function $\bigQ_{\mathrm{H},t}, \bigR_\mathrm{H}$ and noise process parameters, are identified. Furthermore, these parameters remain unaltered when an automation additionally acts upon the control loop. 
\end{assumption}
\begin{remark}
     The parameterset described in Assumption~\ref{assumption1} can be person-specific through personalized identification (e.g. using the approach by~\cite{Karg.2024}) or be generalized over (groups of) persons.
\end{remark}

\subsubsection{Mean and Variance of the Coupled Human-Machine System}\label{sec:MeanVariance}

Based on the calculation of $ \E{\bigx_{t}} $ and $ \bigOmega_{t}^{\bigx}$  of the human movement in Subsection~\ref{sec:humanModel}, we extend the calculations to the joint interaction given in~\eqref{eq:system}: 

\begin{lemma}\label{lemma}
    Let the system dynamics be described by~\eqref{eq:system}, \eqref{eq:sysOutA} and~\eqref{eq:sysOutH}. Let the human control and filter matrices be given by~\eqref{eq:L_H} and~\eqref{eq:K_H}, let Assumption~\ref{assumption1} hold and the automation's control and filter matrices be provided by~\eqref{eq:L_A} and~\eqref{eq:K_A}. Then the joint mean $\E{\bigx_t}$ and variance $\bigOmega_{t}^{\bigx}$ are computed by:  
\begin{align}
    & \left[\begin{array}{l}
        \E{\bigx_{t+1}}\\
        \E{\hat{\bigx}_{\mathrm{H},t+1}}\\
        \E{\hat{\bigx}_{\mathrm{A},t+1}}
    \end{array} \right] = \boldsymbol{\mathcal{A}}_t 
    \left[\begin{array}{l}
        \E{\bigx_{t}}\\
        \E{\hat{\bigx}_{\mathrm{H},t}}\\
        \E{\hat{\bigx}_{\mathrm{A},t}}
    \end{array} \right], \label{eq:meanInteraction} \\
& \mathrm{cov}
\begin{pmatrix}
    \begin{bmatrix}
        \bigx_{t+1} \\
        \hat{\bigx}_{\mathrm{H},t+1}\\
        \hat{\bigx}_{\mathrm{A},t+1}
    \end{bmatrix}
\end{pmatrix}
= \bigmathA_t \; \mathrm{cov}
\begin{pmatrix}
    \begin{bmatrix}
        \bigx_{t} \\
        \hat{\bigx}_{\mathrm{H},t}\\
        \hat{\bigx}_{\mathrm{A},t}
    \end{bmatrix}
\end{pmatrix} \bigmathA_t^{\intercal} \notag \\
&+
\begin{bmatrix}
   \bigOmega^{\bigxi} & \boldsymbol{0} & \boldsymbol{0} \\
    \boldsymbol{0} &\bigK_{\mathrm{H},t}\bigOmega^{\bigomega_{\mathrm{H}}}\bigK_{\mathrm{H},t}^{\intercal} & \boldsymbol{0} \\
    \boldsymbol{0} & \boldsymbol{0} & \boldsymbol{0} \\
\end{bmatrix} \notag \\
&+
\begin{bmatrix}
    \bar{\bigOmega}^{\hat{\bigx}_{\mathrm{H}}}_t & \boldsymbol{0} & \boldsymbol{0} \\
    \boldsymbol{0} & \bar{\bigOmega}^{\bigx}_t & \boldsymbol{0} \\
    \boldsymbol{0} & \boldsymbol{0} & \boldsymbol{0} \label{eq:varianceInteraction} \\
\end{bmatrix}
\end{align}
with 

\begin{align}
    \begin{scriptsize}
    \bigmathA_t \! = \! 
    \begin{+bmatrix}
    \bigA & \! - \! \bigB_{\mathrm{H}}\bigL_{\mathrm{H},t} & \! - \! \bigB_{\mathrm{A}}\bigL_{\mathrm{A},t} \\
   \bigK_{\mathrm{H},t} \bigH_{\mathrm{H}} & \bigA \! - \! \bigB_{\mathrm{H}}\bigL_{\mathrm{H},t} \! - \! \bigK_{\mathrm{H},t} \bigH_{\mathrm{H}} & \boldsymbol{0} \\
   \bigK_{\mathrm{A},t} \bigH_{\mathrm{A}} & \boldsymbol{0} & \bigA \! -\! \bigB_{\mathrm{A}}\bigL_{\mathrm{A},t} \! -\! \bigK_{\mathrm{A},t} \bigH_{\mathrm{A}}\\
   \end{+bmatrix}\!, 
\end{scriptsize} \label{eq:A_t}
\end{align}

\begin{align}
    \bar{\bigOmega}^{\hat{\bigx}_{\mathrm{H}}}_t = &\sum_i \bigC_i\bigL_{\mathrm{H},t} \notag \\ 
    &\Big(\bigOmega^{\hat{\bigx}_{\mathrm{H}}}_t 
    + \mathrm{E}\{\hat{\bigx}_{\mathrm{H},t}\} \mathrm{E}\{\hat{\bigx}_{\mathrm{H},t}\}^\intercal \Big)\bigL_{\mathrm{H},t}^{\intercal} \bigC_i^\intercal , \\
\bar{\bigOmega}^{\bigx}_t = & \sum_i \bigK_{\mathrm{H},t} \bigD_i \Big(\bigOmega^{\bigx}_t + \mathrm{E}\{\bigx_t\} \mathrm{E}\{\bigx_t\}^\intercal \Big)  \bigD_i^\intercal \bigK_{\mathrm{H},t}^\intercal . 
\end{align}
\end{lemma}
\begin{proof}
    See Appendix. 
\end{proof}
These descriptions provide recursive formulas to compute the mean $\E{\bigx_{t}}$ and variance $\bigOmega_{t}^{\bigx}$ of the system state that gets acted upon by a stochastic human and a deterministic automation.

\subsubsection{Optimization-based Automation Parametrization} \label{sec:parametertuning}

Based on Lemma~\ref{lemma}, a controller can be parametrized which explicitly considers its effect on the mean and variance of the coupled human and automation. 

With the previously established foundation, we now introduce our procedure, depicted in Fig.~\ref{fig:procedure}. 
On the basis of identified human parameters, the human-alone mean and variance for a specific movement can be calculated as described in Subsection~\ref{sec:humanModel}. 
The resulting movement is regarded as the natural human behavior.

Using the goals defined in Subsection~\ref{sec:goal} we formulate a parameter optimization problem allowing to shape the resulting mean and covariance of the joint interaction compared to the human-alone action: 

\begin{proposition}
    Let a human-machine system be defined by~\eqref{eq:system}. Let Assumption~\ref{assumption1} hold. Find the optimal automation's cost function parameters $\boldsymbol{q}_{\mathrm{A}}^*$ and $\boldsymbol{r}_{\mathrm{A}}^*$ that result in a control law which minimizes the objective function
\begin{align}
    \min_{\boldsymbol{q}_{\mathrm{A}}, \boldsymbol{r}_{\mathrm{A}}} \; & J_\mathrm{HVROC}(\boldsymbol{q}_{\mathrm{A}}, \boldsymbol{r}_{\mathrm{A}}) = \notag \\
    &s_\mathrm{highMidVar} \left( \bigOmega_{\nicefrac{N}{2},\mathrm{HVROC}}^{\bigp} - \bigOmega_{\nicefrac{N}{2},\mathrm{human}}^{\bigp} \right)^2 \notag \\
    &+s_\mathrm{lowMidVar} \left( \frac{\bigOmega_{\nicefrac{N}{2},\mathrm{HVROC}}^{\bigp}}{\bigOmega_{\nicefrac{N}{2},\mathrm{human}}^{\bigp}} \right)^2  \notag \\ 
    &+s_\mathrm{endVar} \left( \frac{\bigOmega_{N,\mathrm{HVROC}}^{\bigp}}{\bigOmega_{N,\mathrm{human}}^{\bigp}} \right)^2  \notag \\ 
    &+s_\mathrm{ref}\left( \frac{\vert \mathrm{E}\{ \boldsymbol{p}_N \}_\mathrm{HVROC} - \bigp_{\mathrm{ref}} \vert}{\vert \mathrm{E}\{ \boldsymbol{p}_N \}_\mathrm{human} - \bigp_{\mathrm{ref}} \vert} \right)^2  \label{eq:costHVROC} \\
    \mathrm{s.t.} \; &\boldsymbol{q}_{\mathrm{A}} \geq \boldsymbol{0}, \boldsymbol{r}_{\mathrm{A}} > \boldsymbol{0}. \notag
\end{align}
\end{proposition}

\begin{remark}
The four weights $\boldsymbol{s} = \begin{bmatrix} s_\mathrm{highMidVar} & s_\mathrm{lowMidVar} &  s_\mathrm{endVar} & s_\mathrm{ref} \end{bmatrix}$ allow to balance between the four goals given by~\eqref{eq:goal1a}-\eqref{eq:goal3}. 
Only either $s_\mathrm{highMidVar}$ or $s_\mathrm{lowMidVar}$ should be chosen to be non-zero, depending on whether a high or low variability in the task-irrelevant area between goal points is desired. 
\end{remark}
This proposition is solved by a bi-level approach as depicted in Fig.~\ref{fig:procedure}: On the upper level the objective function is evaluated on basis of the mean and variance of the coupled system for a given set of $\boldsymbol{q}_\mathrm{A}$ and $\boldsymbol{r}_\mathrm{A}$, which in turn is calculated by the lower level. 
This results in the optimal parameters $\boldsymbol{q}^*_\mathrm{A}, \boldsymbol{r}^*_\mathrm{A}$, from which the optimal control matrix $\bigL^*_{\mathrm{A},t}$ can be calculated.

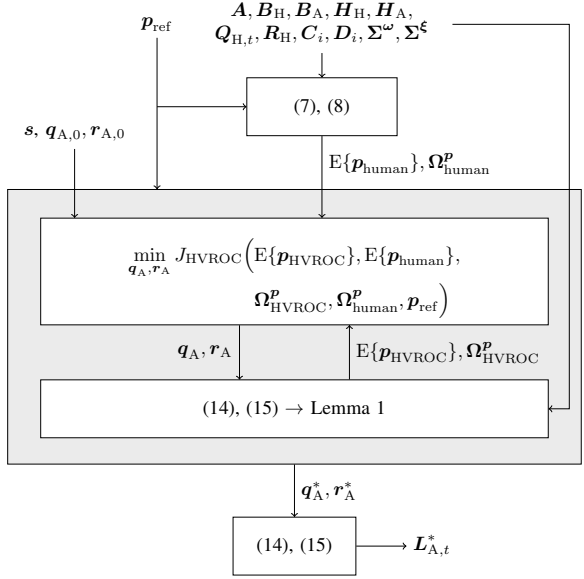
\begin{figure}[thpb]
    \centering
    \resizebox{0.48\textwidth}{!}{
\definecolor{gray1}{RGB}{242,242,242}
\definecolor{gray2}{RGB}{180,180,180}
 \definecolor{shadecolor}{rgb}{0.92,0.92,0.92}

\usetikzlibrary{arrows}

\tikzstyle{block} = [rectangle, draw, fill=white,
   text width=2cm, text centered, minimum height=3em]
\tikzstyle{box} = [rectangle, fill = shadecolor, draw, text width=24em, text centered,  minimum height=8em] 
\tikzstyle{line} = [draw, ->] 

\begin{tikzpicture}[node distance = 2cm, auto]
    
  	\node [text width = 4.5cm, align=center ] (ident) at (0.5,9.5) { $ \boldsymbol{A}, \boldsymbol{B}_\mathrm{H},  \boldsymbol{B}_\mathrm{A},  \boldsymbol{H}_\mathrm{H},  \boldsymbol{H}_\mathrm{A},$ \\ 
  	$\boldsymbol{Q}_{\mathrm{H},t}, \boldsymbol{R}_{\mathrm{H}}, 
  	 \boldsymbol{C}_i, \boldsymbol{D}_i, \boldsymbol{\Sigma}^{\boldsymbol{\omega}}, \boldsymbol{\Sigma}^{\boldsymbol{\xi}}  $}; 
   	\node [block, text width = 2.5cm] (meanHuman) at (0.5,8) { \eqref{eq:E_human}, \eqref{eq:cov_human}};
   	
   	\node [box, inner xsep = 1cm, inner ysep=2.5cm, anchor=north] (box) at (0,6.5) {};

   	\node [block, text width = 9cm] (calcTop) at (0,5) { \begin{align}    	
   	\min_{\boldsymbol{q}_\mathrm{A}, \boldsymbol{r}_\mathrm{A}} J_\mathrm{HVROC} \Big( 
   	& \mathrm{E} \{ \boldsymbol{p}_\mathrm{HVROC} \},  \mathrm{E} \{ \boldsymbol{p}_\mathrm{human} \}, \notag \\
   	& \boldsymbol{\Omega}_{\mathrm{HVROC}}^{\boldsymbol{p}}, \boldsymbol{\Omega}_{\mathrm{human}}^{\boldsymbol{p}}, \boldsymbol{p}_\mathrm{ref}    \Big)   \notag
   	\end{align}};

	\node [block, text width = 9cm] (calcBot) at (0,2.5) { \eqref{eq:L_A}, \eqref{eq:K_A} $\rightarrow $ Lemma~\ref{lemma}}; 

	\node[xshift=-1cm] (anchorTL) at (calcTop.south) {} ;	
	\node[xshift=-1cm] (anchorBL) at (calcBot.north) {};	
	\path[line] (anchorTL.center) -- (anchorBL.center)node[midway, left](q-r) {$\boldsymbol{q}_{\mathrm{A}}, \boldsymbol{r}_{\mathrm{A}} $};

	\node[xshift=1cm] (anchorTR) at (calcTop.south) {};	
	\node[xshift=1cm] (anchorBR) at (calcBot.north) {};	
	\path[line] (anchorBR.center) -- (anchorTR.center)node[midway,  right](q-r) {$ \mathrm{E} \{ \boldsymbol{p}_\mathrm{HVROC} \}, \boldsymbol{\Omega}_{\mathrm{HVROC}}^{\boldsymbol{p}} $}; 

	\node (p_ref) at (-2.5,9.5) {$\boldsymbol{p}_\mathrm{ref}$};
	\node (input) at (-4,7.5) {$\boldsymbol{s}$, $ \boldsymbol{q}_{\mathrm{A},0}, \boldsymbol{r}_{\mathrm{A},0}$
	}; 
	
	\path[line] (input) -- (calcTop.north -| input) {};
	\path[line] (p_ref) -- (box.north -| p_ref) {};
	\path[line]  (meanHuman -| p_ref) -- (meanHuman) {};
	
	\node[block] (output) at (0,0) { \eqref{eq:L_A}, \eqref{eq:K_A} }; %
	
	\path[line] (ident) -- (meanHuman);
	\path[line] (meanHuman) -- (calcTop.north -| meanHuman)node[midway, yshift=0.25cm,](identHuman){$\mathrm{E} \{ \boldsymbol{p}_\mathrm{human} \}, \boldsymbol{\Omega}_{\mathrm{human}}^{\boldsymbol{p}} $};
	\path[line] (box) -- (output)node[midway](out) {$\boldsymbol{q}_{\mathrm{A}}^*, \boldsymbol{r}_{\mathrm{A}}^* $};
		
	\node (inputHelp) at (5,9.5) {}; 	
	\path[line] (ident) -- (inputHelp.center) -- (inputHelp |- calcBot.east) -- (calcBot.east);
	
	\node (L_out) at (2.5,0) {$\boldsymbol{L}^*_{\mathrm{A},t}$};
	\path[line] (output) -- (L_out);
	
\end{tikzpicture}
      }
      \caption{Procedure depicting the bi-level optimization structure.}
      \label{fig:procedure}
\end{figure}

\section{Numerical Example} \label{sec:implementation}

We implement our proposed method on two examples, both tackling a two-dimensional P2P movement to be performed jointly by a human and an automation. The first example is inspired by~\cite{Todorov.2002, Karg.2022} and deals with the short-distance manipulation of a human hand. In the second case we simulate a handheld tool which is grasped by the human hand and which is to be moved over a longer distance. 

\subsection{Simulation Parametrization}
The system state is defined as  $\bigx = \begin{bmatrix} 
    p_x & p_y & \Dot{p_x} & \Dot{p_y} & f_{x} & f_{y} & g_{x} & g_{y}
\end{bmatrix}^\intercal $. 
Both dimensions can be described analogously, we therefore only describe the $x$ dimension from here on. Analogously to~\cite{Todorov.2005}, $p_x$ and $\Dot{p}_x$ denote the position and velocity of the human hand, which by itself is modeled as a point mass of $m_\mathrm{hand} = \SI{1}{kg}$, optionally grasping a tool which adds a tool-mass of $m_\mathrm{tool}$ to the total point-mass $m = m_\mathrm{hand} + m_\mathrm{tool}$. $f_x$ describes the force exerted on the hand, which results from the automation's input $u_{\mathrm{A},x}$ and human muscle force on the hand resulting from a second-order linear filter $g_x$ with the human neural activation $u_{\mathrm{H},x}$ as input. The human filter time constants are chosen as $\tau_1 = \tau_2 = \SI{40}{ms}$. With a time discretization of $\Delta t = \SI{10}{ms}$, we receive dynamic and filter equations:
\begin{subequations}
\begin{align}
    p_{x,t+1}&= p_{x,t} + \Delta t \; \Dot{p}_{x,t}, \label{eq:sys_p}\\
    \Dot{p}_{x,t+1}&=\Dot{p}_{x,t} + \frac{\Delta t}{m} f_{x,t}, \label{eq:sys_v} \\
     f_{x,t+1} &= \left( 1 - \frac{\Delta t}{\tau_2} \right) f_{x,t} + \frac{\Delta t}{\tau_2} g_{x,t} + u_{\mathrm{A},x,t}, \label{eq:sys_f}\\
      g_{x,t+1} &= \left( 1 - \frac{\Delta t}{\tau_1} \right) g_{x,t} + \frac{\Delta t}{\tau_1} u_{\mathrm{H},x,t}. \label{eq:sys_g}
\end{align}
\end{subequations}
The equations for the  $y$-dimension can be described analogously. 
Using~\eqref{eq:sys_p}-\eqref{eq:sys_g}, the system matrices $\bigA$, $\bigB_\mathrm{A}$ and $\bigB_\mathrm{H}$ can be derived. 
Both the human and the automation output equations provide insight into the first six system states: $\bigH_\mathrm{H} = \bigH_\mathrm{A} = \begin{bmatrix} \boldsymbol{I}_{ 6 \times 6} & \boldsymbol{0}_{6 \times 2} \end{bmatrix}$. 

We set the human cost function parameters to $\bigQ_{\mathrm{H},N} = \mathrm{diag}(\begin{bmatrix} 1 & 1 & 0.04 & 0.04 & 0.0004 & 0.0004 & 0 & 0 \end{bmatrix}^\intercal)$ and $\bigR_{\mathrm{H}} = \mathrm{diag}(\begin{bmatrix}  0.000005 & 0.000005 \end{bmatrix}^\intercal)$ in close analogy to \cite{Karg.2022}. 

The scaling matrices for the additive noise processes are chosen as $\boldsymbol{\Sigma}^{\bigomega} = \mathrm{diag} \left( \begin{bmatrix} 0.02 & 0.02 & 0.2 & 0.2 & 1 & 1 \end{bmatrix}^\intercal \right)$ and $\boldsymbol{\Sigma}^{\bigxi} = \boldsymbol{0}$. Furthermore, we define the scaling matrices for the signal-dependent noise processes as $\bigC_1 = \sigma^{\bigu} \bigB$, $\bigC_2 = \sigma^{\bigu} \bigB \begin{bmatrix} -\bigpsi_2 & \bigpsi_1 \end{bmatrix}$ with $\sigma^{\bigu}=0.5$ and $\bigD = \boldsymbol{0}$, with $\bigpsi$ denoting the standard unit vector of $\mathbb{R}^2$.  

The task is to move the human hand from the starting point $p_{x,0} = p_{y,0} = \SI{0}{m}$ to the end point $\bigp_{\mathrm{ref}}$ within $N$ timesteps. For the human hand example we set $p_{x,\mathrm{ref}} = p_{y,\mathrm{ref}} = \SI{0.1}{m}$, $m = m_\mathrm{hand} = \SI{1}{kg}$ and $N=42$. Within the tool manipulation example, the goal point $p_{x,\mathrm{ref}} = p_{y,\mathrm{ref}} = \SI{0.5}{m}$ is to be reached within $N=96$ timesteps with a mass of $m = m_\mathrm{hand} + m_\mathrm{tool} = \SI{10}{kg}$. 


The weights of our introduced cost function~\eqref{eq:costHVROC} allow to influence the shape of the resulting behavior regarding mean accuracy and variability in the goal point as well as variability midway between start and end point.   
In our example we aim at balancing the priorization between all three goals, while either shaping the joint variance to be similar to the human-alone variance or low at $t = \nicefrac{N}{2}$:  $\boldsymbol{s}_\mathrm{high var} = \begin{bmatrix} 1 & 1 & 10 & 0 \end{bmatrix}^\intercal$ and $\boldsymbol{s}_\mathrm{low var} = \begin{bmatrix} 1 & 1 & 0 & 1 \end{bmatrix}^\intercal$. 

We optimize for the first six weights of $\bigQ_\mathrm{A}$, thus for $\boldsymbol{q}(1:6)$, as the automation has no controllability over the neural activation of the system. 
The optimization algorithm starts with initial parameters that are adopted from the human model: $\boldsymbol{q}_0= \begin{bmatrix} 1 & 1 & 0.04 & 0.04 & 0.0004 & 0.0004 \end{bmatrix}^\intercal$ and $\boldsymbol{r}_0 = \begin{bmatrix} 0.000005 & 0.000005 \end{bmatrix}^\intercal $. 

\subsection{Results}

In this subsection we present the simulation results of our proposed \textit{HVROC} in two parametrizations (aiming for either a) a variance that is  similar to the \textit{human}-alone or b) a reduced variance for $t=\nicefrac{N}{2}$) and a linear-quadratic riccati (\textit{LQR}) controller that is parameterized similar to the human as a benchmark example, both in a closed-loop interaction with a simulated human. The simulated \textit{human}-alone behavior is shown for comparison. We look at two cases: the former deals with jointly manipulating a human hand and the latter with the manipulation of a handheld tool. 
 
The \textit{LQR} benchmark controller obtains a state estimation through a Kalman filter. 
The cost function parameters are chosen similarly to the human, with $\bigQ_{\mathrm{LQR}} = \mathrm{diag}(\begin{bmatrix} 1 & 1 & 0.04 & 0.04 & 0.0004 & 0.0004 & 0 & 0 \end{bmatrix}^\intercal)$ for $t= 1 \dots N$ and $\bigR_{\mathrm{LQR}} = \mathrm{diag}(\begin{bmatrix}  0.002 & 0.002 \end{bmatrix}^\intercal)$. 
The resulting joint behavior of the \textit{LQR} together with the simulated human is depicted in black in Fig.~\ref{fig:sim}.


\begin{figure}[h]
    \centering
    \begin{subfigure}[b]{0.475\textwidth}
        \centering
    \resizebox{3.5in}{!}{
%
%
\begin{tikzpicture}

\begin{axis}[%
width=3.044in,
height=3.044in,
at={(0.016in,0.406in)},
scale only axis,
xmin=0,
xmax=42,
xtick={0,10,20,30,40},
xlabel style={font=\color{white!15!black}},
xlabel={Time Steps $t$},
ymin=0,
ymax=0.11,
ytick={  0, 0.1},
ylabel style={font=\color{white!15!black}},
ylabel={E\{$ p_x$\} [$\si{m}$]},
axis background/.style={fill=white},
legend style={at={(0.97,0.03)}, anchor=south east, legend cell align=left, align=left, draw=white!15!black}
]
\addplot [color=red, line width=1.5pt]
  table[row sep=crcr]{%
1	0\\
2	0\\
3	0\\
4	0\\
5	5.59211212138388e-05\\
6	0.000246495726334693\\
7	0.000653236759819613\\
8	0.00134904344707812\\
9	0.00239238534491977\\
10	0.00382499929177087\\
11	0.00567172205704975\\
12	0.00794155321378419\\
13	0.0106293632403218\\
14	0.0137178779703596\\
15	0.017179714989012\\
16	0.0209793423757342\\
17	0.0250748868223482\\
18	0.0294197439155017\\
19	0.0339639642073042\\
20	0.0386554169555779\\
21	0.0434407492389395\\
22	0.0482661690064786\\
23	0.0530780948080908\\
24	0.057823730668681\\
25	0.0624516397550478\\
26	0.0669124010023579\\
27	0.0711594302083726\\
28	0.0751500210473039\\
29	0.0788466058842289\\
30	0.0822181570237544\\
31	0.085241565942425\\
32	0.0879027738509271\\
33	0.0901974335489929\\
34	0.0921310091522648\\
35	0.0937183331598806\\
36	0.0949827335805147\\
37	0.0959550972065412\\
38	0.096673030581396\\
39	0.0971797593611696\\
40	0.0975226181278802\\
41	0.0977484537807249\\
42	0.0978972576360758\\
};
\addlegendentry{Human}

\addplot [color=blue, line width=1.5pt]
  table[row sep=crcr]{%
1	0\\
2	0\\
3	0\\
4	0.000173768304220472\\
5	0.000683340317127598\\
6	0.00166756185218549\\
7	0.00323699535906196\\
8	0.00541664848686698\\
9	0.00823153178171179\\
10	0.011625417917633\\
11	0.0155128315460071\\
12	0.0197968155841297\\
13	0.0243796782107093\\
14	0.0291693286669784\\
15	0.0340825667166743\\
16	0.0390463649784466\\
17	0.0439978961131831\\
18	0.048883826966211\\
19	0.0536592472128032\\
20	0.0582864751072915\\
21	0.0627338944215354\\
22	0.066974941131081\\
23	0.0709872891308957\\
24	0.0747522241116914\\
25	0.0782542112111682\\
26	0.0814806885448375\\
27	0.0844221199291596\\
28	0.0870723033586118\\
29	0.0894288498076343\\
30	0.0914936725545547\\
31	0.0932733247993308\\
32	0.094779072439055\\
33	0.0960266524845834\\
34	0.0970357445661105\\
35	0.0978292262634884\\
36	0.0984323200884513\\
37	0.0988718110226534\\
38	0.0991753713179742\\
39	0.0993708984709871\\
40	0.0994858928929011\\
41	0.0995456287919089\\
42	0.0995706887604917\\
};
\addlegendentry{HVROC - high var}

\addplot [color=blue, dashed, line width=1.5pt]
  table[row sep=crcr]{%
1	0\\
2	0\\
3	0\\
4	9.89307920793721e-05\\
5	0.00041903952393496\\
6	0.00108132460443249\\
7	0.00219191006666538\\
8	0.00380228053091077\\
9	0.00595466729799208\\
10	0.00863479270417562\\
11	0.0118018021956808\\
12	0.0153994234544624\\
13	0.019363450622284\\
14	0.0236268038914196\\
15	0.0281228075420431\\
16	0.0327872067946526\\
17	0.0375593154750004\\
18	0.0423825736643828\\
19	0.0472047235785056\\
20	0.0519777541500592\\
21	0.0566577226649825\\
22	0.0612045572575331\\
23	0.0655819000809014\\
24	0.0697569993996751\\
25	0.0737006740280303\\
26	0.0773873996577365\\
27	0.0807955670122189\\
28	0.0839079291623268\\
29	0.086712187512065\\
30	0.0892015972198209\\
31	0.0913754469818088\\
32	0.0932392798858805\\
33	0.0948047664287121\\
34	0.0960892306675019\\
35	0.097114891821243\\
36	0.0979079345131613\\
37	0.098497635786773\\
38	0.0989156085874668\\
39	0.0991949915793578\\
40	0.099369565807748\\
41	0.0994711416206144\\
42	0.0995260331214027\\
};
\addlegendentry{HVROC - low var}

\addplot [color=black, line width=1.5pt]
  table[row sep=crcr]{%
1	0\\
2	0\\
3	0\\
4	0.000160776601208743\\
5	0.000607548437614762\\
6	0.00143890432847065\\
7	0.00273022715532772\\
8	0.00448914923058673\\
9	0.00673597759578837\\
10	0.0094267611583738\\
11	0.0124986600541397\\
12	0.0158836207056859\\
13	0.0195159072736041\\
14	0.0233360074281092\\
15	0.0272921569173035\\
16	0.0313403837114881\\
17	0.0354437059499642\\
18	0.0395709115789852\\
19	0.0436951990002937\\
20	0.0477928606634806\\
21	0.0518421126889901\\
22	0.0558220690555467\\
23	0.0597118304152342\\
24	0.0634897477396826\\
25	0.0671329602916472\\
26	0.0706172961603029\\
27	0.073917615609213\\
28	0.0770086453422313\\
29	0.0798662631590253\\
30	0.0824690835064294\\
31	0.0848001273852305\\
32	0.0868483336689981\\
33	0.0886096841230968\\
34	0.0900878090508508\\
35	0.0912940350203642\\
36	0.0922469377933815\\
37	0.0929716621635991\\
38	0.0934990904331184\\
39	0.093864656617981\\
40	0.0941068458109622\\
41	0.094263473546483\\
42	0.0943669067277766\\
};
\addlegendentry{LQR}

\end{axis}

\begin{axis}[%
width=3.044in,
height=3.044in,
at={(4.253in,0.406in)},
scale only axis,
xmin=0,
xmax=42,
xtick={0,10,20,30,40},
xlabel style={font=\color{white!15!black}},
xlabel={Time Steps $t$},
ymin=0,
ymax=3.5e-05,
ytick={    0, 1e-05, 2e-05, 3e-05},
ylabel style={font=\color{white!15!black}},
ylabel={cov($ p_x$) [$\si{m^2}$]},
axis background/.style={fill=white}
]
\addplot [color=red, line width=1.5pt, forget plot]
  table[row sep=crcr]{%
1	0\\
2	0\\
3	0\\
4	0\\
5	1.56358589890643e-09\\
6	2.04428152051469e-08\\
7	1.09234727747626e-07\\
8	3.68942915782852e-07\\
9	9.36306166601743e-07\\
10	1.95185625432392e-06\\
11	3.52207892543726e-06\\
12	5.68830659235383e-06\\
13	8.41003756838957e-06\\
14	1.15649014176129e-05\\
15	1.49632679540518e-05\\
16	1.83730234310116e-05\\
17	2.15491270516847e-05\\
18	2.42628550108893e-05\\
19	2.63267648187383e-05\\
20	2.76128122516359e-05\\
21	2.80623588874864e-05\\
22	2.76879489443399e-05\\
23	2.65676396578566e-05\\
24	2.48333120977352e-05\\
25	2.26547854438992e-05\\
26	2.02217316743245e-05\\
27	1.77253678634669e-05\\
28	1.53417154430526e-05\\
29	1.32178772322692e-05\\
30	1.14623180697544e-05\\
31	1.01396063057888e-05\\
32	9.26958626215119e-06\\
33	8.83058352176849e-06\\
34	8.76605712579578e-06\\
35	8.99407948750512e-06\\
36	9.418847620551e-06\\
37	9.94341603647838e-06\\
38	1.04830632893767e-05\\
39	1.09775781730156e-05\\
40	1.13996338787725e-05\\
41	1.17528372410907e-05\\
42	1.20589388304043e-05\\
};
\addplot [color=blue, line width=1.5pt, forget plot]
  table[row sep=crcr]{%
1	0\\
2	0\\
3	0\\
4	0\\
5	1.56358589890643e-09\\
6	2.04428152051469e-08\\
7	1.09234727747626e-07\\
8	3.68289731604705e-07\\
9	9.28398652376796e-07\\
10	1.91461533334136e-06\\
11	3.41595194156767e-06\\
12	5.46880088220104e-06\\
13	8.04822177910865e-06\\
14	1.10664229286358e-05\\
15	1.43803748107208e-05\\
16	1.78074198767198e-05\\
17	2.1145571227808e-05\\
18	2.41948519230265e-05\\
19	2.67767337910249e-05\\
20	2.87497027084175e-05\\
21	3.00198350553468e-05\\
22	3.05461042436956e-05\\
23	3.03407744439301e-05\\
24	2.94654402271971e-05\\
25	2.80232893255273e-05\\
26	2.61482715824448e-05\\
27	2.39920627303377e-05\\
28	2.17099366869822e-05\\
29	1.94468543985676e-05\\
30	1.732522418427e-05\\
31	1.54358162445505e-05\\
32	1.38330696278167e-05\\
33	1.25353804974515e-05\\
34	1.15299906496074e-05\\
35	1.07812598513715e-05\\
36	1.02406137274204e-05\\
37	9.85617281664416e-06\\
38	9.5805179130134e-06\\
39	9.3760272722512e-06\\
40	9.21734051018973e-06\\
41	9.09040333607037e-06\\
42	8.98934364943858e-06\\
};
\addplot [color=blue, dashed, line width=1.5pt, forget plot]
  table[row sep=crcr]{%
1	0\\
2	0\\
3	0\\
4	0\\
5	1.56358589890643e-09\\
6	2.04428152051469e-08\\
7	1.09234727747626e-07\\
8	3.68164214947203e-07\\
9	9.2672504592784e-07\\
10	1.90487207544056e-06\\
11	3.37592506491108e-06\\
12	5.34241257789604e-06\\
13	7.72929128861811e-06\\
14	1.03949310595002e-05\\
15	1.31548157499339e-05\\
16	1.58104018776612e-05\\
17	1.81763580771193e-05\\
18	2.010173363163e-05\\
19	2.14831308457692e-05\\
20	2.22698952641729e-05\\
21	2.2462492699664e-05\\
22	2.21058553826286e-05\\
23	2.12796784164911e-05\\
24	2.00874109783955e-05\\
25	1.86452185405217e-05\\
26	1.707174939435e-05\\
27	1.54792073121989e-05\\
28	1.39659997390602e-05\\
29	1.26110862255093e-05\\
30	1.14701112000113e-05\\
31	1.05734588480728e-05\\
32	9.92642864021833e-06\\
33	9.5116632593571e-06\\
34	9.29371794193876e-06\\
35	9.22535968415138e-06\\
36	9.25486828780541e-06\\
37	9.33336377798576e-06\\
38	9.42124422833681e-06\\
39	9.49294592653739e-06\\
40	9.53911276545977e-06\\
41	9.56376863863276e-06\\
42	9.57785557131989e-06\\
};
\addplot [color=black, line width=1.5pt, forget plot]
  table[row sep=crcr]{%
1	0\\
2	0\\
3	0\\
4	0\\
5	1.56358589890643e-09\\
6	2.04428152051469e-08\\
7	1.09234727747626e-07\\
8	3.68168205961438e-07\\
9	9.27281234762562e-07\\
10	1.91145135516548e-06\\
11	3.40879056772933e-06\\
12	5.44572300598659e-06\\
13	7.971517264504e-06\\
14	1.08585040489373e-05\\
15	1.39185138947682e-05\\
16	1.69300852596759e-05\\
17	1.96691223012474e-05\\
18	2.19367434617754e-05\\
19	2.35803454289946e-05\\
20	2.45061905521173e-05\\
21	2.46835600377287e-05\\
22	2.41416559739151e-05\\
23	2.29610081359834e-05\\
24	2.12612696854193e-05\\
25	1.91871658830581e-05\\
26	1.68941221397472e-05\\
27	1.45348163394074e-05\\
28	1.22476180977682e-05\\
29	1.01476002925824e-05\\
30	8.3205270594167e-06\\
31	6.81994313078142e-06\\
32	5.66723633799641e-06\\
33	4.85432827713827e-06\\
34	4.3485000949636e-06\\
35	4.09884369534424e-06\\
36	4.04377316168245e-06\\
37	4.11902716122121e-06\\
38	4.2658257824129e-06\\
39	4.43830588859619e-06\\
40	4.6086675958643e-06\\
41	4.76725122206518e-06\\
42	4.91678351217044e-06\\
};
\end{axis}

\begin{axis}[%
width=7.812in,
height=3.125in,
at={(0in,0in)},
scale only axis,
xmin=0,
xmax=1,
ymin=0,
ymax=1,
axis line style={draw=none},
ticks=none,
axis x line*=bottom,
axis y line*=left
]
\end{axis}
\end{tikzpicture}%
}
        \caption{Example~1: Manipulation of a human hand with $m = \SI{1}{kg}$, $p_{x,\mathrm{ref}} = \SI{0.1}{m}$, $N~=~42$. }
    \label{fig:simHand}
    \end{subfigure}
    \begin{subfigure}[b]{0.475\textwidth}
        \centering
    \resizebox{3.5in}{!}{
    \input{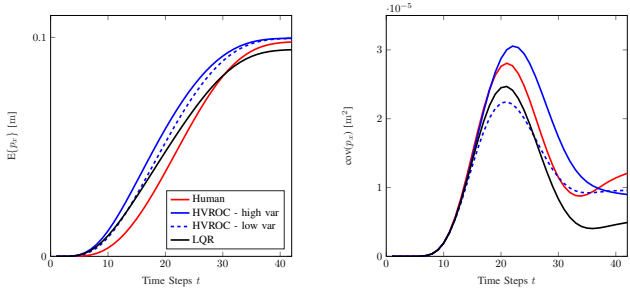}
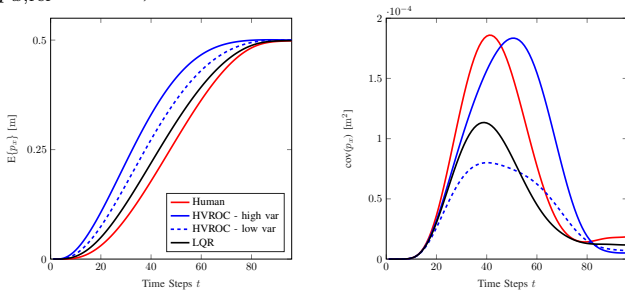}
    \caption{Example~2: Manipulation of a handheld tool with $m = \SI{10}{kg}$, $p_{x,\mathrm{ref}} = \SI{0.5}{m}$, $N=96$. }
    \label{fig:simTool}
    \end{subfigure}
    \caption{Mean and variance of the position $p_x$. \textit{Human}-alone behavior is depicted in red. The joint behavior of our proposed \textit{HVROC} with the simulated human is shown in blue; solid and dotted lines represent a high and low desired variability in the task-irrelevant area, respectively. Black depicts the joint behavior of the \textit{LQR} and human.}
    \label{fig:sim}
\end{figure}

\begin{table}[htbp]
    \caption{Simulation results comparing the \textit{HVROC} in two parametrizations with the \textit{human}-alone and the \textit{LQR} in two example scenarios.}
\begin{center}
      \begin{tabular}{| c || c | c | c | c |}
\hline
\makecell[l]{} & \makecell{ \textbf{Human} \\ alone } & \makecell{\textbf{HVROC}  \\ high var} & \makecell{\textbf{HVROC}  \\ low var} & \makecell{ \textbf{LQR} }\\
\hline
 \multicolumn{5}{|c|}{Example~1: Human hand} \\
\hline
\makecell{$\mathrm{cov} ( p_{x,N/2} )$ \\ in $ 10^{-5}  \si{m^2}$} &  $2.8 $  & $3.0$ & $2.3$ & $2.5$  \\
\hline
\makecell{ $\mathrm{cov} ( p_{x,N} )$ \\ in $ 10^{-6} \si{m^2}$} & $12.1$   & $9.0$ & $9.6$ & $4.9$ \\
\hline
 \makecell{ $|\mathrm{E} \{ p_{x,N} \} - p_{x,\mathrm{ref}} | $ \\ in $\si{mm}$ } &  $ 2.1$  & $0.4$ & $0.5$ & $5.6$\\  
 \hline
 \makecell{ $t(|\mathrm{E} \{ p_{x,t} \} - p_{x,\mathrm{ref}} | $ \\ $ < 0.05 \cdot p_{x,\mathrm{ref}})$} & $37$ & $33$ & $34$ & - \\
\hline
\multicolumn{5}{|c|}{Example~2: Handheld tool manipulation} \\
\hline
\makecell{$\mathrm{cov} ( p_{x,N/2} )$ \\ in $ 10^{-5}  \si{m^2}$} &  $16.6 $  & $18.1$ & $7.6$ & $9.3$  \\
\hline
\makecell{ $\mathrm{cov} ( p_{x,N} )$ \\ in $ 10^{-6} \si{m^2}$} & $1.9$   & $0.5$ & $0.7$ & $1.2$ \\
\hline
 \makecell{ $|\mathrm{E} \{ p_{x,N} \} - p_{x,\mathrm{ref}} | $ \\ in $\si{mm}$ } &  $ 15.0$  & $0.03$ & $0.05$ & $11.0$\\  
\hline
\makecell{ $t(|\mathrm{E} \{ p_{x,t} \} - p_{x,\mathrm{ref}} | $ \\ $ < 0.05 \cdot p_{x,\mathrm{ref}})$} & $78$ & $63$ & $71$ & $76$ \\
\hline
    \end{tabular}
    \label{table:results}
\end{center}
\end{table}

The joint behavior of the \textit{HVROC} interacting with the human is depicted in blue in Fig.~\ref{fig:sim}. Solid and dotted lines represent a high and low desired variability in the task-irrelevant area, specified through either weighing $s_\mathrm{highMidVar}$ or $s_\mathrm{lowMidVar}$. Quantitative results are provided in Tbl.~\ref{table:results}. 

Regarding the variability in the task-irrelevant area it can be observed that the joint variance behaves as intended: a high variance between the start and goal point, similar to the \textit{human}-alone, can be identified in the solid line. A restricted variance in $t = \nicefrac{N}{2}$, compared both to \textit{human}-alone and \textit{LQR}, shows for the dotted line. 

Looking at the goal point variability, our approach leads to a reduced position variance at $t = N$ compared to the \textit{human}-alone for both examples and also compared to the \textit{LQR} in Example~2. 

Furthermore, the \textit{HVROC} exhibits better mean performance in reaching the goal point concerning accuracy and time. Looking at Example~1, \textit{HVROC} accuracy outperforms the \textit{human}-alone by factor $4$ and the \textit{LQR} by factor $10$. The improvement is even more visible for Example~2, where the \textit{HVROC} approaches reach the goal point with a remaining error of $|\mathrm{E} \{ p_{x,N} \} - p_{x,\mathrm{ref}} | = \SI{0.03}{mm} / \SI{0.05}{mm}$ compared to $\SI{15.0}{mm}$ and $\SI{11.0}{mm}$ for \textit{human}-alone and \textit{LQR}, respectively. Additionally, all \textit{HVROC} results reach and stay within a $5 \% $ error bandwidth around the goal point clearly faster then the \textit{LQR}, with the latter even showing a stationary error $> 5 \% $ for Example~1. 

In our examples both dimensions are parameterized identical, therefore the same results can be observed for the $y$-dimension. 

\subsection{Discussion}

Our introduced \textit{HVROC} allows to find automation parameters that optimize the joint behavior of a deterministic automation and a stochastic human. 
Using the introduced objective function, the variance of the human-machine interaction can be shaped as desired, e.g. to restrict the overall variance or keep it similar to the human natural behavior of being elevated in task-irrelevant areas. 
Further work may analyze the effect of different variability patterns on the human experience. 

Additionally, the joint mean behavior of the overall system is improved using the \textit{HVROC}. 
As mentioned by~\cite{Todorov.2005, Karg.2022}, the human noise parameters influence its mean behavior, which regular controller designs do not take into consideration. 
Our approach however, explicitly takes into account the effect which the human noise parameters exert on the joint mean, leading to an improved accuracy and improved speed in reaching the goal point. 
It must be noted that the presented choice of objective function weights only represent one sample choice. Other weights may be chosen according to the desired importance on accuracy or variability patterns. Furthermore, other objective functions that consider the joint behavior may be defined. 

The comparison to the \textit{LQR} shows the need for our optimally and variability-respecting parameterized \textit{HVROC}: the \textit{LQR} shows worse overall performance in accuracy and speed and a strong intervention in human natural variability.



\section{Conclusion}

In this paper we propose a novel control approach for physical Human-Machine Interaction. State-of-the-art solutions are often 
model-based, and therefore generalizable, 
however, they are exclusively based on deterministic models of the human behavior. This results in a mismatch to neuroscientific literature which models the human stochastically. Based on these stochastic human models, we introduce formulations that allow the computation of the mean and variance of a coupled human-machine system. The resulting formulas are then used as a basis for an optimization-based parametrization of the automation.
The benefits of using the stochastic model are twofold: For once, the overall performance of the joint interaction is enhanced, since the effect of human noise processes on the mean behavior is incorporated in the control design. 
Additionally, our controller design allows the shaping of the resulting joint variance: The human natural variability can be adjusted depending on the user's preference.
\section*{APPENDIX: Proof of Lemma 1}

\begin{proof}\label{proof}
    Extending \eqref{eq:E_human} by the automations systems estimation~\eqref{eq:E_aut} and substituting~\eqref{eq:u_aut} and~\eqref{eq:sysOutA}, \eqref{eq:meanInteraction} results.  

    Calculating the variance of the extended system state including $\bigx_{t+1}$, $\hat{\bigx}_{\mathrm{H},t+1}$ and $\hat{\bigx}_{\mathrm{A},t+1}$, we derive:
    \begin{align}
        & \mathrm{cov}
\begin{pmatrix}
    \begin{bmatrix}
        \bigx_{t+1} \\
        \hat{\bigx}_{\mathrm{H},t+1}\\
        \hat{\bigx}_{\mathrm{A},t+1}
    \end{bmatrix}
\end{pmatrix}
= \bigmathA_t \; \mathrm{cov}
\begin{pmatrix}
    \begin{bmatrix}
        \bigx_{t} \\
        \hat{\bigx}_{\mathrm{H},t}\\
        \hat{\bigx}_{\mathrm{A},t}
    \end{bmatrix}
\end{pmatrix} \bigmathA_t^{\intercal} \notag \\
&+
\mathrm{E}
\begin{Bmatrix}
    \begin{bmatrix}
        - \sum_i \varepsilon_t^{(i)} \bigC_i \bigL_{\mathrm{H},t} \hat{\bigx}_{\mathrm{H},t}\\
        \bigK_{\mathrm{H},t} \sum_i \epsilon_t^{(i)} \bigD_i \bigx_t \\
        \boldsymbol{0} \\
    \end{bmatrix}
    \begin{bmatrix}
        - \sum_i \varepsilon_t^{(i)} \bigC_i \bigL_{\mathrm{H},t} \hat{\bigx}_{\mathrm{H},t}\\
        \bigK_{\mathrm{H},t} \sum_i \epsilon_t^{(i)} \bigD_i \bigx_t \\
        \boldsymbol{0} \\
    \end{bmatrix}^\intercal
\end{Bmatrix}\notag \\
&+ \mathrm{E}
    \begin{Bmatrix}
    \begin{bmatrix}
        \bigxi_t\\
       \bigK_{\mathrm{H},t} \bigomega_{\mathrm{H},t} \\
       \bigK_{\mathrm{A},t} \bigomega_{\mathrm{A},t} \\
    \end{bmatrix}
    \begin{bmatrix}
        \bigxi_t\\
       \bigK_{\mathrm{H},t} \bigomega_{\mathrm{H},t} \\
       \bigK_{\mathrm{A},t} \bigomega_{\mathrm{A},t} \\
    \end{bmatrix}^\intercal \label{eq:proof_1}
\end{Bmatrix}
\end{align}
by exploiting the zero-mean and independence to each other of $\bigvarepsilon_t$, $\bigepsilon_t$, $\bigxi_t$, $\bigomega_{\mathrm{H},t}$ and $\bigomega_{\mathrm{A},t}$ as well as their independence to $\bigx_{t}$, $\hat{\bigx}_{\mathrm{H},t}$ and $\hat{\bigx}_{\mathrm{A},t}$. Analogously to~\cite{Karg.2022} we can further simplify each element of the matrix of the second summand of~\eqref{eq:proof_1} due to 
above mentioned independencies as well as $\E{\varepsilon_t^{(i)} \varepsilon_t^{(j)}} = \E{\epsilon_t^{(i)} \epsilon_t^{(j)}} = \delta_{ij}$ ($\delta_{ij} = 1$ for $i=j$, $\delta_{ij} = 0$ for $i\neq j$) and $\E{\varepsilon_t^{(i)} \epsilon_t^{(j)}} = 0$ ($\forall i,j$): 

\begin{align}
    &\E{\sum_{i} \varepsilon_t^{(i)} \bigC_i \bigL_t \hat{\bigx}_{\mathrm{H},t} \hat{\bigx}_{\mathrm{H},t}^{\intercal} \bigL_t^{\intercal} \sum_{j} \bigC_j^{\intercal} \varepsilon_t^{(j)}} \notag \\
    &= \sum_{i} \bigC_i \bigL_t \E{\hat{\bigx}_{\mathrm{H},t} \hat{\bigx}_{\mathrm{H},t}^{\intercal}} \bigL_t^{\intercal} \bigC_i^{\intercal} \label{eq:proof_2} \\
    &\E{-\sum_{i} \varepsilon_t^{(i)} \bigC_i \bigL_t \hat{\bigx}_{\mathrm{H},t} \bigx_t^{\intercal} \sum_{j} \bigD_j^{\intercal} \epsilon_t^{(j)} \bigK_t^{\intercal}} = \boldsymbol{0} \\
    &\E{-\bigK_t \sum_{i} \epsilon_t^{(i)} \bigD_i \bigx_t \hat{\bigx}_{\mathrm{H},t}^{\intercal} \bigL_t \sum_{j} \bigC_j^{\intercal} \varepsilon_t^{(j)}} = \boldsymbol{0} \\
    &\E{\bigK_t \sum_{i} \epsilon_t^{(i)} \bigD_i \bigx_t \bigx_t^{\intercal} \sum_{j} \bigD_j^{\intercal} \epsilon_t^{(j)} \bigK_t^{\intercal}} \notag \\ 
    &= \sum_i \bigK_t \bigD_i \E{ \bigx_t \bigx_t^{\intercal} } \bigD_i^{\intercal} \bigK_t^{\intercal} \label{eq:proof_3}.
\end{align}
We apply  $\E{\hat{\bigx}_{\mathrm{H},t} \hat{\bigx}_{\mathrm{H},t}^{\intercal}} = \bigOmega_{t}^{\hat{\bigx}_{\mathrm{H}}} + \E{\hat{\bigx}_{\mathrm{H},t}}\E{\hat{\bigx}_{\mathrm{H},t}}^{\intercal}$ in \eqref{eq:proof_2} and $\E{ \bigx_t \bigx_t^{\intercal} } = \bigOmega_{t}^{\bigx} + \E{\bigx_t}\E{\bigx_t}^{\intercal}$ in \eqref{eq:proof_3}, resulting in
\begin{align}
    \begin{bmatrix}
        \bar{\bigOmega}^{\hat{\bigx}_{\mathrm{H}}}_t & \boldsymbol{0} & \boldsymbol{0} \\
        \boldsymbol{0} & \bar{\bigOmega}^{\bigx}_t & \boldsymbol{0} \\
        \boldsymbol{0} & \boldsymbol{0} & \boldsymbol{0}\\
    \end{bmatrix}
\end{align}
for the second summand of~\eqref{eq:proof_1}.
Due to the independence of $\bigxi_t$, $\bigomega_{\mathrm{H},t}$ and $\bigomega_{\mathrm{H},t}$, the third summand reduces to
\begin{align}
    \begin{bmatrix}
       \bigOmega^{\bigxi} & \boldsymbol{0} & \boldsymbol{0} \\
        \boldsymbol{0} &\bigK_{\mathrm{H},t}\bigOmega^{\bigomega_{\mathrm{H}}}\bigK_{\mathrm{H},t}^{\intercal} & \boldsymbol{0} \\
        \boldsymbol{0} & \boldsymbol{0} & \boldsymbol{0} \\
    \end{bmatrix},
    \end{align}
leading to the simplification of~\eqref{eq:proof_1} to \eqref{eq:varianceInteraction}.
\end{proof}

\addtolength{\textheight}{-14cm}   






\bibliography{1_tex/bibliography.bib}
\bibliographystyle{IEEEtran}


\end{document}